\documentclass[letterpaper, 10 pt, conference]{ieeeconf}

\IEEEoverridecommandlockouts                              
\makeatletter
\def\ps@headings{%
\def\@oddhead{\mbox{}\scriptsize\rightmark \hfil \thepage}%
\def\@evenhead{\scriptsize\thepage \hfil \leftmark\mbox{}}%
\def\@oddfoot{}%
\def\@evenfoot{}}

\makeatother

\usepackage{graphicx}
\usepackage[noadjust]{cite}
\usepackage{tabularx,ragged2e,booktabs,caption}
\usepackage{times}
\usepackage{alltt}
\usepackage{verbatim}
\usepackage{moreverb}
\usepackage{amsmath}
\usepackage{amssymb}
\usepackage{url}
\usepackage{epsfig}
\usepackage{epstopdf}
\usepackage{subfigure}
\usepackage{multirow}
\usepackage[ruled,vlined,linesnumbered]{algorithm2e}
\usepackage{mathrsfs}
\usepackage{listings}
\usepackage[none]{hyphenat}
\usepackage{xcolor}
\usepackage{cases}

\newtheorem{theorem}{Theorem}

\newtheorem{definition}{Definition}

\newtheorem{lemma}{Lemma}

\begin{document}

\title{\LARGE \bf An Ordinary Differential Equation Framework for Stability Analysis of Networks with Finite Buffers}

\author{Xinyu~Wu,~\IEEEmembership{Student Member,~IEEE,}
        Dan~Wu,~\IEEEmembership{Member,~IEEE,}
        Eytan~Modiano,~\IEEEmembership{Fellow,~IEEE}}

%
%


%

\maketitle
\thispagestyle{empty}
\pagestyle{empty}

\begin{abstract}
We consider the problem of network stability in finite-buffer systems. We observe that finite buffer may affect stability even in simplest network structure, and we propose an ordinary differential equation (ODE) model to capture the queuing dynamics and analyze the stability in buffered communication networks with general topology. For single-commodity systems, we propose a sufficient condition, which follows the fundamental idea of backpressure, for local transmission policies to stabilize the networks based on ODE stability theory.
We further extend the condition to multi-commodity systems, with an additional restriction on the coupling level between different commodities, which can model networks with per-commodity buffers and shared buffers. The framework characterizes a set of policies that can stabilize buffered networks, and is useful for analyzing the effect of finite buffers on network stability.
\end{abstract}

\section{Introduction}
\label{sec:introduction}

Network overloading occurs ubiquitously due to burst flow injection, transmission link failures, or node malfunction. A typical feature of network overloading is queue backlog accumulation, which often results in large queuing delay, network congestion and performance degradation. Therefore, a core engineering problem is to design transmission policies to avoid network overloading, which is commonly referred to as \emph{stabilizing} the network.

In their seminal paper on network stability, Tassiulas and Ephremides showed that backpressure routing can stabilize networks whenever the packet arrival rates are within the network stability region \cite{tassiulas1990stability}. Their result elegantly solves the network stability problem for systems with unbounded buffers, and serves as a basic policy design framework for a large body of works that capture utility maximization \cite{li2014receiver,neely2008fairness,huang2012lifo,yu2018new}, delay minimization \cite{cui2015enhancing,huang2012lifo,alresaini2015backpressure}, and network fairness \cite{georgiadis2006optimal,neely2008fairness,chan2010fairness,li2014dynamic}.

However, most of the related works rely on the assumption of unbounded buffers, which deviates from the fact that in reality buffers are finite \cite{giaccone2007throughput}. In practice, internal nodes in a communication network often have limited buffers \cite{Juniper,baron2009capacity}. For example, on-chip networks have very small internal buffers due to area and power limitation, and similarly, satellite networks have small buffers on-board the satellite.  In constrast, buffers of the source nodes of the arriving packets have sufficient capacity 
to absorb bursty packet arrivals  \cite{le2010optimal}, e.g. in a satellite network the buffer at the ground terminal can be relatively large. Therefore in this paper, we assume that the internal buffers are finite while the source buffers are unbounded, {i.e., sufficiently large to not saturate. 
}

The existence of finite buffers gives rise to many additional design issues, including whether buffer overflow is permitted, whether the buffer is shared by different commodities, and whether the queue backlog of different commodities in the buffer is controllable. {In this paper, we consider general buffer settings, while we do not allow buffer overflow, i.e., packet transmission to a saturated node. This is desirable in practice since it prevents packet dropping and significantly reduces the retransmission delay due to buffer overflow \cite{le2010optimal}. Therefore, network stability is guaranteed as long as the queue backlog at all source nodes, which have unbounded buffers, is bounded.} 

The introduction of finite buffers may induce critical difference in network performance. Consider the toy example in Fig.~\ref{fig:TwoNodeToyExample}, showing that a policy may no longer stabilize the network when  buffers are finite. The system transmits two commodities with arrival rate $\lambda_1$ and $\lambda_2$, and destination node $T_1$ and $T_2$ respectively. We assume that a policy based on the principal idea of backpressure is applied\footnote{Details of this policy will be presented as \eqref{eqn:DiffBacklogBP-Single} in Section \ref{subsec:single_commodity_stability}.}, where link $(\ell,K)$ transmits with rate equal to  the capacity value $c_{\ell K}$ when the queue backlog in node $\ell$ is longer than in node $K$ ($\ell=1,2$), while otherwise it does not transmit, in any time unit. The buffer at node $K$ is shared by the two commodities. 
We observe that when $b_K$, the buffer size of node $K$, is infinite, then this policy can stabilize commodity 2 for $\lambda_2\in [0,3]$ when $\lambda_1>2$. However, when $b_K$ is finite and small enough such that it can be saturated, for example $b_K=6$ in Fig.~\ref{fig:TwoNodeToyExample}, then this policy can  stabilize commodity 2 only for $\lambda_2\in \left[0,1.5\right]$\footnote{Brief explanation is in the caption of Fig.~\ref{fig:TwoNodeToyExample}, and derivation of the result is deferred to Section \ref{subsec:shared-buffer}.}. 
This unveils that finite buffer affects stability even in the simplest one-hop network structure, giving rise to the need to study the stability of the queue dynamics of general buffered systems.

\begin{figure}[!htbp]
\centering
\includegraphics[width=0.9\linewidth]{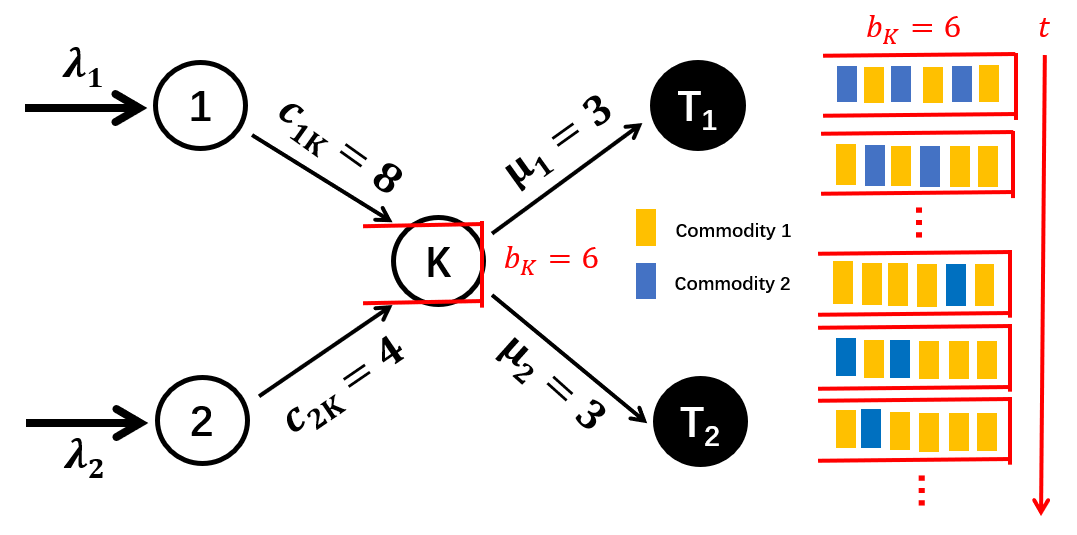}
\caption{\small Finite buffer may affect stability result. On the right is an example of the queue dynamics in node $K$ following the backpressure policy. 
Due to its higher injection rate into the buffer of $K$, commodity 1 takes up higher ratio in node $K$ and squeezes out commodity 2 under backpressure, and in the final state we can show that average number of commodity 2 packets in the buffer of node $K$ is $1.5$, which arises from $1.5=(\mu_1/c_{1K})\times c_{2K}$, with details deferred to Section \ref{subsec:shared-buffer}. Therefore the actual throughput of commodity 2 is $1.5$, less than $\mu_2=3$, due to the finite buffer.}
\label{fig:TwoNodeToyExample}
\end{figure}

A plethora of works have tried to incorporate finite buffers in network analysis. Giaccone et. al. studied the throughput region of finite-buffer network systems and discussed the relationship between buffer size and throughput under two proposed policies \cite{giaccone2007throughput}. Le et. al. studied the relationship between buffer size and network utility under a modified backpressure mechanism \cite{le2010optimal}, and Lien et. al. designed a dynamic algorithm to stabilize any admissible traffic conditioned on a finite internal buffer with size larger than a certain bound \cite{lien2011maximizing}. All of these works propose \emph{specific} policies and analyze their performance on finite-buffer systems, under \emph{certain assumptions} such as deterministic routing \cite{giaccone2007throughput}, separate buffers for different commodities \cite{giaccone2007throughput,le2010optimal}, equal buffer sizes \cite{giaccone2007throughput,le2010optimal}, and minimum buffer size requirements \cite{lien2011maximizing}. Thus, a systematic study of buffered networks with \emph{general buffer size setting} and \emph{general control policies} is still lacking.


In this paper, we propose an ordinary differential equation (ODE) model to analyze network stability of a general buffered system. Unlike prior works that study stochastic dynamics, our approach uses deterministic dynamics which  
simplifies the analysis while still capturing the key aspects of finite-buffer systems and their impact on network stability. Our main goal is to show that the ODE framework serves as a promising approach for analyzing finite-buffer systems,  {which elegantly models transmission policies under arbitrary buffer settings}. To that end, we are able to 
obtain more general results as compared to prior works on finite-buffer system stability.

{Based on the ODE framework,} our contributions include: (i) For single-commodity systems, we derive a sufficient condition for a set of local policies to stabilize the network based on ODE stability theory, and our condition has similar physical intuition to backpressure; (ii) For multi-commodity systems, we similarly derive a sufficient condition for network stability, with an additional condition that captures the coupling level between different commodities. The core idea is that these conditions reduce the network stability problem to a problem of testing the existence of an equilibrium point, and thus facilitate stability analysis. The existence or lack of equilibrium points in multi-commodity systems can explain the effect of finite buffer on stability shown in Fig.~\ref{fig:TwoNodeToyExample}.


Our framework resembles the fluid model \cite{dai1995positive,dai2005maximum} which was proposed as a flow-based approximation to the discrete network systems to obtain results for throughput \cite{shah2011fluid} and delay \cite{markakis2018delay}. However, the fluid model captures the scaled limit of the queue backlogs which for nodes with finite buffers is not very meaningful. A more closely related framework is the ODE model used to study the Transmission Control Protocol (TCP) \cite{misra1999stochastic,gu2004integrating}. While sharing similar queue dynamics modeling with   \cite{misra1999stochastic,gu2004integrating}, our approach can capture more general policies.



\section{Single-Commodity System}

In this section, we introduce the ODE model for single-commodity systems in buffered communication networks, and utilize ODE stability theory to study network stability. Specifically, we reveal a general sufficient condition for local policies to stabilize the systems. 
{The results, with explicit guidance for transmission policy design in practice, serve as the basis for our analysis of multi-commodity systems and are the main technical contribution of this paper.}

\subsection{Basic Setting}
\label{subsec:basic-setting}

Given an acyclic directed network $\mathcal{G}=(\mathcal{V},\mathcal{E})$ with $|\mathcal{V}|=N$ nodes and $|\mathcal{E}|=M$ links. Each node $i$ has a queue buffer, whose size is denoted by $b_i$. The queue length at node $i$ at time $t$ is denoted by $q_i(t)$, where $q_i(t)\in [0,b_i], \forall t$.
We use an $N\times 1$ vector $\mathbf{q}(t)$ to denote the queue length vector of the system, and {denote $\mathcal{Q}:=\times_{i=1}^N [0,b_i]$ as the set of feasible queue length vectors, where $\times_{i=1}^{N}$ denotes the $N$-dimensional Cartesian product.}
Packets in the buffer of node $i$ can be transmitted to an adjacent downstream node $j$ through link $(i,j)\in \mathcal{E}$. The transmission rate on link $(i,j)$ at time $t$, denoted by $g_{ij}(t)$, captures the number of packets transmitted over $(i,j)$ in a time unit. Each link $(i,j)$ is associated with a capacity value $c_{ij}$, which is the largest flow that link $(i,j)$ can transmit at any time. Specifically, $0\leq g_{ij}(t) \leq c_{ij},~\forall (i,j)\in \mathcal{E}$.

In communication systems, the transmission rate $g_{ij}(t)$ on each link $(i,j)$ is generally determined by the controller at node $i$, according to the queue length vector $\mathbf{q}(t)$ and network configurations, including link capacity values $\{c_{ij}\}_{(i,j)\in \mathcal{E}}$ and node buffer size values $\{b_i\}_{i\in \mathcal{V}}$. Therefore $g_{ij}(t)$ is also referred to as the \emph{transmission policy} over link $(i,j)$. In this paper, we consider a set of \emph{local}, {\emph{stationary}} policies that \emph{do not allow buffer overflow}. (i) Locality: {We say that a policy is local if $g_{ij}(t)$ depends only on
$q_i(t)$ and $q_j(t)$.} {Local policies are attractive due to their simple implementation and light communication overhead for information exchange. }
This definition is to some extent extremely local as in reality, node $i$ can have information on all $q_k(t)$'s for $(i,k)\in \mathcal{E}$.
However, we show that even with this restricted queue information, a policy can stabilize the network under certain general conditions {with clear physical intuition}. {(ii) Stationarity: A policy is stationary if it does not depend on time explicitly. 
Note that stationary policies can depend on the network state (e.g., queue size, channel conditions); i.e., a local policy $g_{ij}(t)$ can be denoted by
$g_{ij}(q_i(t),q_j(t))$\footnote{A local policy $g_{ij}(t)$ should also depend on $c_{ij}$, $b_i$ and $b_j$, but for
brevity we neglect them in the notation as they do not vary with time.}. We neglect the notation $t$ in the following
unless specified.} (iii) No buffer overflow: Any link $(i,j)$ must stop packet transmission once the buffer of node $j$ is saturated, i.e., $g_{ij}=0$ when $q_j=b_j$. In addition to the above three constraints on the policy, we naturally have $g_{ij}=0$ when $q_i=0$, which means link $(i,j)$ has nothing to transmit when the buffer of upstream node $i$ is empty. {For technical convenience, we assume $g_{ij}(q_i,q_j)$ is first-order differentiable with respect to $q_i$ and $q_j$, and we show that it can well approximate discrete policies in Section \ref{subsec:single_commodity_stability}.}

Packets are injected into the networks {at their source nodes}. We denote the packet injection rate at node $i$ as $\lambda_i(t)$ and assume that $\lambda_i(t)$ is independent from the queue length vector $\mathbf{q}(t)$. 
Packets may depart from the networks at any node, and we model this by introducing a meta destination node $T$ to receive the departing packets. The transmission rate from node $i$ to $T$ is $g_{iT}(q_i)$, purely based on $q_i$ under the local policy, and the capacity is denoted by $c_{iT}:=\mu_i$, where $\mu_i$ denotes the maximum departure rate at node $i$. 



We now formulate the ODE to characterize the queue dynamics. The general form is given by,
\begin{equation}
\label{eqn:ODE}
\dot{\mathbf{q}}=\frac{d\mathbf{q}}{dt}:=\mathbf{f}(\mathbf{q}),
\end{equation}
where $\mathbf{f}(\mathbf{q}):=[f_i(\mathbf{q})]_{i=1}^N\in \mathbb{R}^{N}$ denotes the system dynamics.
Due to flow conservation at each node, under the local policy we define earlier, we have for $\forall i\in \mathcal{V}$,
\begin{equation}
\label{eqn:local-policy}
\small
f_i(\mathbf{q})=\lambda_i(t)+\sum_{k:(k,i)\in \mathcal{E}} g_{ki}(q_k,q_i)-\sum_{j:(i,j)\in \mathcal{E}} g_{ij}(q_i,q_j)-g_{iT}(q_i).
\end{equation}
In the above dynamics, the term $\lambda_i(t)+\sum_{k:(k,i)\in \mathcal{E}} g_{ki}(q_k,q_i)$ denotes the total packet inflow rate to node $i$, while the term $\sum_{j:(i,j)\in \mathcal{E}} g_{ij}(q_i,q_j)+g_{iT}(q_i)$ denotes the total packet departure rate from node $i$. 

\subsection{Stability Analysis}
\label{subsec:single_commodity_stability}

{
Next, we study network stability of the dynamics \eqref{eqn:local-policy} under local policies. We consider the queue length stability, as given by Definition \ref{def:stable}.

\begin{definition}
\label{def:stable}
The system \eqref{eqn:local-policy} is queue length stable if $\forall i\in \mathcal{V}$, $\lim_{t\rightarrow \infty} \sup q_i(t)<\infty$.
\end{definition}

Queue length stability ensures the boundedness of the queue backlog at each node. In a system with finite buffers, where overflow is not permitted, instability can only occur at the nodes with unbounded buffers (i.e., source nodes). We note that the assumption of ``unbounded'' buffers at the source node captures the reality that in many systems internal buffers are small, and buffers at source nodes are relatively large. It is also a useful modeling tool that captures the impact of finite buffers on congestion by ``pushing'' congestion to the source nodes. Finally, we note that it can be easily shown that queue stability at the source nodes implies rate stability, which is a more meaningful notion of stability in finite-buffer systems\footnote{Rate stability implies that the arrival rate is equal to the departure rate \cite{neely2010stability}.}.}

We study the queue length stability based on the stability analysis of the \emph{equilibrium points} of ODE system.
\begin{definition}
$\mathbf{q}^*$ is an equilibrium point of the system \eqref{eqn:local-policy} if $\mathbf{f}(\mathbf{q}^*)=\boldsymbol{0}$.
\end{definition}

The stability of an equilibrium point $\mathbf{q}^*$ is defined based on either a Lyapunov function or the Jacobian matrix. In this paper, we take the latter way to define stability.
\begin{definition}
The equilibrium point $\mathbf{q}^*$ of an ODE system \eqref{eqn:local-policy} is asymptotically stable if all the eigenvalues of the Jacobian matrix $\mathbf{J}$ at $\mathbf{q^*}$, i.e.,
$$
\mathbf{J}\bigg|_{\mathbf{q}=\mathbf{q}^*}=\left[\frac{\partial f_i(\mathbf{q})}{\partial q_j}\right]\bigg|_{\mathbf{q}=\mathbf{q}^*,~i,j=1,2,\dots,N}
$$
have {negative real parts.}
\end{definition}
%



An equilibrium point of an ODE being stable means that the dynamics will not move away from this equilibrium under small disturbances, which is different from queue length stability.
In the following, our goal is to connect the notion of equilibrium point stability to queue length stability. The intuitive idea is that if the system in \eqref{eqn:ODE} has a unique asymptotically stable equilibrium point, then queue length stability follows as the dynamics will not diverge to infinity.
{
Specifically, we first seek a sufficient condition for a policy $g_{ij}(q_i,q_j)$ such that any equilibrium point $\mathbf{q}^*$ is asymptotically stable, which ensures local queue length stability. We then extend it to a global condition which  guarantees the global queue length stability. 

}
%


\subsubsection{Local Result}
We first derive a sufficient condition of the local policy such that an equilibrium point of the system \eqref{eqn:local-policy} is asymptotically stable.


\begin{theorem}
\label{prop:single1}
For a local policy, if there exists an equilibrium point $\mathbf{q}^*$, such that at $\mathbf{q}=\mathbf{q}^*$,
\begin{equation}
\label{eqn:local_stability}
\begin{cases}
\frac{\partial g_{ij}(q_i,q_j)}{\partial q_i} > 0,~\frac{\partial g_{ij}(q_i,q_j)}{\partial q_j} < 0,~\forall (i,j)\in \mathcal{E}, \\
\frac{\partial g_{iT}(q_i)}{\partial q_i} > 0,~\exists i\in \mathcal{V}
\end{cases}
\end{equation}
then the ODE is asymptotically stable at $\mathbf{q}^*$.
\end{theorem}

\begin{proof}
\emph{(Sketch)} According to \eqref{eqn:local-policy}, given any node $u\in \mathcal{V}$, and for $\forall i=1,\dots,N$,
\begin{equation}
\label{eqn:Jacobian_single}
\begin{aligned}
\mathbf{J}_{iu}
=\begin{cases}
0,~(i,u)\notin \mathcal{E}, (u,i)\notin \mathcal{E}, i\neq u \\
-\frac{\partial g_{iu}(q_i,q_u)}{\partial q_u},~(i,u)\in \mathcal{E}  \\
\frac{\partial g_{ui}(q_u,q_i)}{\partial q_u},~(u,i)\in \mathcal{E}  \\
\sum_{k:(k,u)\in \mathcal{E}} \frac{\partial g_{ku}(q_k,q_u)}{\partial q_u}
\\ \quad - \sum_{j:(u,j)\in \mathcal{E}} \frac{\partial g_{uj}(q_u,q_j)}{\partial q_u} - \frac{\partial g_{uT}(q_u)}{\partial q_u},~i=u \\
\end{cases}
\end{aligned}
\end{equation}

Under the condition \eqref{eqn:local_stability}, we can verify that
all the diagonal entries of $\mathbf{J}$ are negative, and $|\mathbf{J}_{ii}|\geq \sum_{u:u\neq i}|\mathbf{J}_{ki}|,~\forall i\in \mathcal{V}$. Hence $\mathbf{J}$ is a column diagonally dominant matrix. This indicates that all the eigenvalues have non-positive real parts by the Gershgorin circle theorem \cite{li2019eigenvalue}, and starting from this fact, we can further prove that all eigenvalues have negative real parts under \eqref{eqn:local_stability}, with details deferred in Appendix \ref{subsec:proof_theorem_1}. 
\end{proof}

Theorem \ref{prop:single1} conveys that under any local policy so that the two conditions in \eqref{eqn:local_stability} hold, the network will be queue length stable given the initial queue length lies in a sufficiently small neighborhood of the equilibrium point. The first condition is related to the intuition of the backpressure policy \cite{tassiulas1990stability}: The fluid flows from high pressure nodes to low pressure nodes. The queue length represents the pressure, and once the pressure at the upstream node $i$ increases, then pressure difference increases for $i$ and $j$ thus $g_{ij}(q_i,q_j)$, the flow over $(i,j)$, should be larger; in contrast, once the pressure at the downstream node $j$ increases, the difference decreases and thus $g_{ij}(q_i,q_j)$ should be smaller. {The second condition means there exists at least one node (egress node) where packets can depart from the network }, and the departure rate increases as more packets accumulate in the buffer. We show later in this section specific examples of network policies, including backpressure, that satisfy both conditions.

Different from prior works proposing and proving queue length stability under a \emph{specific} policy, Theorem \ref{prop:single1} presents an explicit and intuitive condition which \emph{generalizes} a set of local policies that can guarantee an equilibrium point to be stable, with no limitations over the buffer setting embedded in the policy function $g_{ij}(q_i,q_j),~\forall (i,j)\in \mathcal{E}$.


\subsubsection{Global Result} {Theorem \ref{prop:single1} can only guarantee local stability. In Theorem \ref{prop:single2} we extend to a global stability result, by identifying a sufficient condition for the local policy to have at most one globally asymptotically stable equilibrium point, which is a crucial step towards queue length stability.}

\begin{theorem}
\label{prop:single2}
If for all queue length vector $\mathbf{q}\in \mathcal{Q}$, the policy $g_{ij}(q_i,q_j)$ for any $(i,j)\in \mathcal{E}$ satisfies \eqref{eqn:local_stability}, 
then the system \eqref{eqn:local-policy} either has a unique asymptotically stable equilibrium point or does not have any equilibrium points.
\end{theorem}

\begin{proof}
Suppose there exists more than one equilibrium point, then the ODE system must have some equilibrium point $\tilde{\mathbf{q}}$ that is not asymptotically stable. However, by Theorem \ref{prop:single1}, since \eqref{eqn:local_stability} holds at $\tilde{\mathbf{q}}$, then $\tilde{\mathbf{q}}$ should be asymptotically stable, which is a contradiction.
\end{proof}

{
Theorem \ref{prop:single2} ensures that any equilibrium point is unique and asymptotically stable when all feasible queue length vectors satisfy \eqref{eqn:local_stability} under a local policy. This indicates that under such a policy, the problem to determine the queue length stability of the system can be reduced to determine whether there exists any feasible solution to the equation $\mathbf{f}(\mathbf{q})=\mathbf{0}$. This reduces the network stability problem to a feasibility test problem of $N$-dim equations, which facilitates network stability analysis especially under large $N$.}

To interpret Theorem \ref{prop:single2} more concretely, we propose the following policy examples that globally satisfy \eqref{eqn:local_stability}.

\textbf{Policy based on Backpressure:} At time $t$, each node $i$ compares its queue length to the queue length of each of its downstream nodes $j$. If $q_i(t)>q_j(t)$, link $(i,j)$ transmits with its capacity value $c_{ij}$, and otherwise it does not transmit. The policy can be formulated using differentiable rate functions as follows.
\begin{equation}
\label{eqn:DiffBacklogBP-Single}
\begin{cases}
g_{ij}(q_i,q_j)=\frac{1}{1+e^{-a(q_i-q_j-\epsilon)}}\frac{1}{1+e^{-a(b_j-\epsilon-q_j)}}c_{ij},~\forall (i,j)\in \mathcal{E} \\
g_{iT}(q_i)=\frac{1}{1+e^{-a(q_i-\epsilon)}}\mu_i,~\text{if $i$ is an egress node}
\end{cases}
\end{equation}
where $a>0$ and $\epsilon>0$ are preset values. It is easy to verify that \eqref{eqn:DiffBacklogBP-Single} satisfies \eqref{eqn:local_stability} globally for any feasible $\mathbf{q}$. The form of \eqref{eqn:DiffBacklogBP-Single} matches backpressure under $a\rightarrow \infty$ and  $\epsilon:=1/\sqrt{a}\rightarrow 0$, which transmits the packets from $i$ to $j$ with rate $c_{ij}$ if and only if $q_i>q_j$ and $q_j<b_j$, i.e., the buffer of $j$ is not saturated. This shows that \eqref{eqn:DiffBacklogBP-Single} is an approximation to backpressure under sufficiently large $a$ and small $\epsilon$. {Moreover, the form of $g_{iT}(q_i)$ guarantees work-conservation under $a\rightarrow \infty$ and  $\epsilon:=1/\sqrt{a}\rightarrow 0$, i.e., maximum departure rate if there are packets in the buffer.}


\vspace{1mm}
\textbf{Policy based on Buffer Occupancy Level:} Consider
$$
\begin{cases}
g_{ij}(q_i,q_j)=\frac{1}{1+e^{-a(q_i-\epsilon)}}\left(1-\frac{q_j}{b_j}\right)c_{ij},~\forall (i,j)\in \mathcal{E}\\
g_{iT}(q_i)=\frac{1}{1+e^{-a(q_i-\epsilon)}}\mu_i, ~\text{if $i$ is an egress node}
\end{cases}
$$
We can similarly verify it globally satisfies \eqref{eqn:local_stability}. Taking $a$ large enough and $\epsilon \rightarrow 0$ to guarantee $g_{ij}=0$ when $q_i=0$. The transmission rate of link $(i,j)$ in this policy declines linearly with respect to $q_j/b_j$, the buffer occupancy level of node $j$. This policy does not transmit with rate equal to link capacity when the downstream node's buffer is not empty, {but Theorem \ref{prop:single2} reveals that it can also stabilize the network if there exists an equilibrium point for \eqref{eqn:local-policy} under this policy.}

\subsection{Existence of Equilibrium Point}
\label{subsec:equilibrium-existence}

With Theorem \ref{prop:single2}, The remaining gap to proving queue length stability for a local policy that satisfies \eqref{eqn:local_stability} globally is to show that there exists an equilibrium point for \eqref{eqn:local-policy} under this policy. There is no well-known answer to this question except verifying the feasibility of \eqref{eqn:local-policy} directly. We identify one sufficient condition for the existence of an equilibrium point in Lemma \ref{prop:single3}, proved by Poincare-Miranda Theorem, a multi-dimensional version of the intermediate value theorem. 

%

{\small \textbf{Poincare-Miranda Theorem} \cite{kulpa1997poincare}:
Consider $n$ continuous functions of $n$ variables, $g_{1},\ldots ,g_{n}$. Assume that for each variable $x_{i}$, the function $g_{i}$ is constantly negative when $x_{i}=-1$ and constantly positive when $x_{i}=1$. Then there is a point in the $n$-dimensional cube $[-1,1]^n$ such that $g_{1},\ldots ,g_{n}$ are simultaneously equal to $0$.}

\begin{lemma}
\label{prop:single3}
Suppose that there exists finite values $\{\bar{b}_j\}_{i=1}^N$ and $\{\underline{b}_j\}_{i=1}^N$ such that for every node $i$ and any $q_j\in[\underline{b}_j,\bar{b}_j],~\forall j\neq i$: (i) when $q_i=\bar{b}_i$, the policy leads to $f_i(\mathbf{q})\leq 0$; (ii) when $q_i=\underline{b}_i$, $f_i(\mathbf{q})\geq 0$, then system \eqref{eqn:local-policy} has an equilibrium point in $\bar{\mathcal{B}}:=\times_{i=1}^N [\underline{b}_i,\bar{b}_i]$.
\end{lemma}

\begin{proof}
We can apply the Poincare-Miranda Theorem over the cube $\bar{\mathcal{B}}:=\times_{i=1}^N [\underline{b}_i,\bar{b}_i]$ and taking $g_i:=f_i$ in \eqref{eqn:ODE}, which ensures at least one $\mathbf{q}$ such that $\mathbf{f}(\mathbf{q})=\mathbf{0}$.\footnote{In the proof, we do not follow the condition \emph{constantly negative/positive}, instead we consider \emph{constantly non-positive/non-negative}. This does not affect our result as we consider the existence of solution in a closed cube.}
\end{proof}

Lemma \ref{prop:single3} is a general result {for the existence of an equilibrium point} in which $\bar{\mathcal{B}}$ requires specification for a particular policy. We consider the policy \eqref{eqn:DiffBacklogBP-Single} as an example to show how we can obtain $\bar{\mathcal{B}}$ in Appendix \ref{sec:lemma}. Combining Theorem \ref{prop:single2} and Lemma \ref{prop:single3}, we have the following theorem for queue length stability.




\begin{theorem}
\label{thm:single-commodity}
For any local policy that satisfies the conditions in both Theorem \ref{prop:single2} and Lemma \ref{prop:single3}, there exists a unique stable equilibrium point in $\bar{\mathcal{B}}$ defined in Lemma \ref{prop:single3}, and the system is queue length stable (and thus rate stable).
\end{theorem}

\begin{proof}
Theorem \ref{prop:single2} ensures there exists at most one stable equilibrium point, while Lemma \ref{prop:single3} ensures there exists at least one equilibrium point. Therefore there exists a unique stable equilibrium point and thus starting at any queue length vector, the dynamics will converge to the equilibrium.
\end{proof}

Theorem \ref{thm:single-commodity} can be viewed as an extension to Theorem \ref{prop:single2} which only adds a sufficient condition for equilibrium point existence. The above stability results differ from previous works in that (i) they capture a set of policies, and (ii) they can be applied in systems with arbitrary buffer settings.

\section{Multi-Commodity System}
\label{sec:MultiCommodityStability}

We extend the above results to multi-commodity systems, where different commodities are coupled due to shared links or buffers. We identify a sufficient condition for queue length stability that is similar to the single-commodity case, but with an additional condition on the coupling among different commodities.

\subsection{Basic Setting}

Suppose that the system consists of $C$ commodities.
We use $q_i^{(\ell)}(t)$ to denote the queue length of commodity $\ell$ at node $i$ and time $t$, 
and an $N_{\ell}\times 1$ vector $\mathbf{q}^{(\ell)}(t)$ to denote the queue length vector for commodity $\ell$, where $N_{\ell}$ denotes the number of nodes on the available paths of commodity $\ell$. We denote vector $\mathbf{q}(t)$ that concatenates $\{\mathbf{q}^{(\ell)}(t)\}_{\ell=1}^C$ as the queue length vector for the entire network.
For commodity $\ell$, we denote the arrival rate at node $i$ as $\lambda_i^{(\ell)}(t)$, the transmission rate on link $(i,j)$ as $g_{ij}^{(\ell)}(t)$, and the departure rate to outside the networks at node $i$ as $g_{iT_{\ell}}^{(\ell)}(t)$, where $T_{\ell}$ denotes the meta destination connected for commodity $\ell$.




We also consider local, stationary policies in the multi-commodity case with no overflow permitted, namely the same conditions as for the single-commodity systems.
The queueing dynamics under a local policy for any commodity $\ell$ at any node $i$ is given by
\begin{equation}
\label{eqn:multi-equation}
\begin{aligned}
\dot{q}_i^{(\ell)}&=\lambda_i^{(\ell)}+\sum_{k:(k,i)\in \mathcal{E}} g_{ki}^{(\ell)}(\{q_k^{(p)}\}_{p=1 }^{C},\{q_i^{(p)}\}_{p=1}^{C})
\\&-\sum_{j:(i,j)\in \mathcal{E}} g_{ij}^{(\ell)}(\{q_i^{(p)}\}_{p=1}^C,\{q_j^{(p)}\}_{p=1}^C)
-g_{iT}^{(\ell)}(\{q_i^{(p)}\}_{p=1}^{C}).
\end{aligned}
\end{equation}
%

\subsection{Stability Analysis}
\label{sec:stability-multi}

We follow a similar pattern to the single-commodity case. We first derive conditions for local queue length stability, and extend to a global sufficient condition that ensures an equilibrium point $\mathbf{q}^*$ to be globally unique and stable, which captures a set of local policies that can achieve queue length stability for all commodities. 

To study the stability at $\mathbf{q}^*$, we need to first introduce an important concept of block diagonally dominant matrix.
\begin{definition}
A matrix $\mathbf{J}$ is called a \emph{block diagonally dominant matrix} if $\mathbf{J}$ can be partitioned into the following $C\times C$ blocks
\begin{equation}
\label{eqn:blockmatrix}
\mathbf{J}=
\begin{bmatrix}
\mathbf{J}_{1,1} & \mathbf{J}_{1,2} & \cdots & \mathbf{J}_{1,C} \\
\mathbf{J}_{2,1} & \mathbf{J}_{2,2} & \cdots & \mathbf{J}_{2,C} \\
\vdots & \vdots & \ddots & \vdots \\
\mathbf{J}_{C,1} & \mathbf{J}_{C,2} & \cdots & \mathbf{J}_{C,C}
\end{bmatrix}
\end{equation}
where $\mathbf{J}_{i,j}\in \mathbb{R}^{N_i\times N_j}$, the diagonal submatrices $\{\mathbf{J}_{i,i}\}_{i=1}^C$ are nonsingular, and for some operator norm $||\cdot||$,
$||\mathbf{J}_{j,j}^{-1}||^{-1}\geq \sum_{k=1,k\neq j}^{C} ||\mathbf{J}_{j,k}||,~\forall j = 1,\dots,C$. $\mathbf{J}$ is called a block strictly diagonally dominant matrix if the last condition is a strict inequality.
\end{definition}

This definition is directly related to the system \eqref{eqn:multi-equation}:
Each block on the diagonal ($\mathbf{J}_{\ell,\ell}$) represents the derivative of the dynamics of commodity $\ell$ with respect to the $q_i^{(\ell)}$ at each node $i$, while each off-diagonal block ($\mathbf{J}_{\ell^{\prime},\ell}$) denotes its derivative with respect to $q_i^{(\ell^{\prime})}$ for commodity $\ell^{\prime}\neq \ell$, which reflects the coupling between commodities $\ell$ and $\ell^{\prime}$. Intuitively, block diagonally dominance requires the total coupling effect of each commodity $\ell$ (i.e. $\sum_{k=1,k\neq \ell}^{C} ||\mathbf{J}_{\ell,k}||$)  to be relatively small.
%

Before proving our results based on the block diagonally dominant matrix, we need to introduce the concepts of M-matrix and absolute norm, and a theorem \cite{feingold1962block} about matrix eigenvalues.
\begin{definition}
A matrix is called an M-matrix if all of its eigenvalues have nonnegative real parts and all its off-diagonal entries are nonpositive.
\end{definition}

\begin{definition}
A norm $||\cdot||$ is an {absolute norm} if $||\mathbf{q}||=||~|\mathbf{q}|~||$, where $|\mathbf{q}|$ denotes the element-wise absolute value.
\end{definition}

{\small
\textbf{Theorem:} \cite{feingold1962block} Let matrix $\mathbf{J}$ be partitioned as in \eqref{eqn:blockmatrix}, and let $\mathbf{J}$ be block strictly diagonally dominant. Further, suppose that $-\mathbf{J}_{j,j}$ is an M-matrix, $\forall 1\leq j\leq C$, and the norm is an absolute norm. The any eigenvalue of $\mathbf{J}$
has negative real part.}

The theorem informs that we need to ensure three points to prove that an equilibrium point is asymptotically stable: (i) $\mathbf{J}$ is block strictly diagonally dominant; (ii) $-\mathbf{J}_{\ell,\ell}$ is an M-matrix, $\forall \ell = 1,\dots,C$; (iii) the norm is an absolute norm.
The last point obviously holds when we apply the $l_2$-norm. We have the following lemma for (i) and (ii).

\begin{lemma}
\label{prop:multi1}
Suppose that the Jacobian matrix $\mathbf{J}$ of \eqref{eqn:multi-equation} at an equilibrium point $\mathbf{q}^*$ under the local policy satisfies
\begin{itemize}
\item Block strictly diagonal dominance:
$$\sqrt{\lambda_{\min}(\mathbf{J}_{\ell,\ell}\mathbf{J}_{\ell,\ell}^T)}> \sum_{p=1,p\neq \ell}^C\sqrt{\lambda_{\max}(\mathbf{J}_{p,\ell}^T\mathbf{J}_{p,\ell})}$$
for $\forall \ell=1,\dots,C$.
\item M-matrix condition: For $\forall (i,j)\in\mathcal{E}$ and $\forall \ell=1,\dots,C$,
$\frac{\partial g_{ij}^{(\ell)}}{\partial q_i^{(\ell)}}> 0,~\frac{\partial g_{ij}^{(\ell)}}{\partial q_j^{(\ell)}}< 0$, and
$\frac{\partial g_{iT}^{(\ell)}}{\partial q_i^{(\ell)}}> 0,~\exists i\in \mathcal{V}$.
\end{itemize}
then the equilibrium point $\mathbf{q}^*$ is asymptotically stable.
\end{lemma}

\begin{proof}
The block strictly diagonal dominance constraint is derived based on constraint $||\mathbf{J}_{s,s}^{-1}||^{-1}\geq \sum_{t=1,t\neq s}^{C} ||\mathbf{J}_{t,s}||$ under $l_2$-norm. Specifically,
$$
\begin{aligned}
||\mathbf{J}_{s,s}^{-1}||&=\sqrt{\lambda_{\max}((\mathbf{J}_{s,s}^T)^{-1}\mathbf{J}_{s,s}^{-1})}
=\sqrt{\lambda_{\max}(\mathbf{J}_{s,s}\mathbf{J}_{s,s}^T)^{-1}}
\\&=\left(\sqrt{\lambda_{\min}(\mathbf{J}_{s,s}\mathbf{J}_{s,s}^T)}\right)^{-1}.
\end{aligned}
$$
Therefore $||\mathbf{J}_{s,s}^{-1}||^{-1}=\sqrt{\lambda_{\min}(\mathbf{J}_{s,s}\mathbf{J}_{s,s}^T)}$ while the RHS $\sum_{t=1,t\neq s}^{C} ||\mathbf{J}_{t,s}||=\sum_{t=1,t\neq s}^{C} \sqrt{\lambda_{\max}(\mathbf{J}_{t,s}^T\mathbf{J}_{t,s})}$ based on the definition of $l_2$-norm.
The M-matrix condition can be similarly proved as Theorem \ref{prop:single1} to show that under the condition, $-\mathbf{J}_{\ell,\ell}$ is an M-matrix. Then all the eigenvalues of $\mathbf{J}$ have negative real parts based on \cite{feingold1962block}, and thus $\mathbf{q}^*$ is asymptotically stable.
\end{proof}

Compared with Theorem \ref{prop:single1} for single-commodity systems, the \emph{block strictly diagonal dominance} in Lemma \ref{prop:multi1} is an additional condition for the restriction of coupling level among different commodities, while for each commodity, the \emph{M-matrix condition} coincides with the conditions \eqref{eqn:local_stability}.
In fact, in the case that different commodities do not affect each other,
the block strictly diagonal dominance holds naturally as all the off-diagonal blocks are zero matrices, and thus Lemma \ref{prop:multi1} is reduced to a $C$-fold version of Theorem \ref{prop:single1}.



Similar to the single-commodity case, we can obtain a sufficient condition of a local policy such that \eqref{eqn:multi-equation} has a unique stable equilibrium point, reducing the stability problem to testing the existence of an equilibrium point. 

%
%

\begin{theorem}
\label{prop:multi2}
Suppose that the conditions in Lemma \ref{prop:multi1} hold for any feasible $\mathbf{q}$, then there either exists a unique asymptotically stable equilibrium point for \eqref{eqn:multi-equation} or there does not exist any equilibrium point. 
\end{theorem}

\subsection{Existence of Equilibrium Point}
\label{subsec:equilibrium}

In terms of the existence of an equilibrium point, we similarly apply the Poincare-Miranda Theorem. However since it can only capture cube form regions of $\mathbf{q}$, we can only obtain results for systems with per-commodity buffers, where node $i$ allocates a portion of its buffer to each commodity $\ell$, with length denoted as $b_i^{(\ell)}$, which satisfies $q_{i}^{(\ell)}\leq b_i^{(\ell)}$ and $\sum_{p=1}^C b_i^{(p)}\leq b_i$. For systems with shared buffers, the constraint is $\sum_{p=1}^C q_i^{(p)}\leq b_i$, not a cube, hence this theorem is not applicable. Systems with shared buffers will be discussed in Section \ref{subsec:shared-buffer}.

\begin{lemma}
\label{prop:multi3}
For every commodity $\ell$, under the condition that there exists finite values $\{\bar{b}_j^{(\ell)}\}_{i=1}^N$ and $\{\underline{b}_j^{(\ell)}\}_{i=1}^N$ such that for every node $i$, and $q_j^{(\ell)}\in[\underline{b}_j^{(\ell)},\bar{b}_j^{(\ell)}],~\forall j\neq i$: (i) when $q_i^{(\ell)}=\bar{b}_i^{(\ell)}$, $f_i^{(\ell)}(\mathbf{q})\leq 0$; (ii) when $q_i^{(\ell)}=\underline{b}_i^{(\ell)}$, $f_i^{(\ell)}(\mathbf{q})\geq 0$. There exists a feasible equilibrium point $\mathbf{q}\in \bar{\mathcal{B}}:=\times_{i=1}^N \times_{p=1}^C [\underline{b}_i^{(p)},\bar{b}_i^{(p)}]$ for system \eqref{eqn:multi-equation}.

\end{lemma}

\begin{proof}
The proof is a simple extension of Lemma \ref{prop:single3} by applying its idea for each commodity $\ell$.
\end{proof}
%

Note that Lemma \ref{prop:multi3} only decouples different commodities at node buffers, which means it still applies when different commodities affect each other's transmission rates on shared links. Combining Lemma \ref{prop:multi3} with Theorem \ref{prop:multi2}, we obtain the following result regarding queue length stability for all commodities. The proof is similar to Theorem \ref{thm:single-commodity}.

\begin{theorem}
\label{thm:multi}
For policies under the conditions of Theorem \ref{prop:multi2} and Lemma \ref{prop:multi3}, there exists a globally unique stable equilibrium point for \eqref{eqn:multi-equation}, which guarantees queue length stability for all commodities.
\end{theorem}

\subsection{Shared buffer Case Study: Switched Networks}
\label{subsec:shared-buffer}

In networks with shared buffers, one commodity may fully occupy a shared buffer and thus squeeze out other commodities, which may induce the instability of these commodities. While we have not been able to obtain results for general networks under this setting, we can obtain explicit results over single-hop network structure as Fig.~\ref{fig:TwoNodeToyExample}, {which serves as the basic structure for server farms \cite{li2014dynamic} and switches in data center networks \cite{al2008scalable}.}

Consider the policy based on backpressure \eqref{eqn:DiffBacklogBP-Single} for the system in Fig.~\ref{fig:TwoNodeToyExample}. Let
$$
\begin{cases}
\alpha_{ij}^{(\ell)}=\frac{1}{1+e^{-a(q_i^{(\ell)}-q_j^{(\ell)}-\epsilon)}},~\forall (i,j)\in \mathcal{E},~\ell=1,2 \\
\beta_K=\frac{1}{1+e^{-a(b_K-q_K^{(1)}-q_K^{(2)}-\epsilon)}},
\end{cases}
$$
where $a\rightarrow \infty$ and $\epsilon:=1/\sqrt{a}\rightarrow 0$, and we can then write the queueing dynamics as
\begin{subnumcases}
\small
\dot{q}_1^{(1)} =\lambda_1-c_{1K}\alpha_{1K}^{(1)}\beta_K,~\dot{q}_K^{(1)}=c_{1K}\alpha_{1K}^{(1)}\beta_K-\mu_1 \alpha_{KT}^{(1)} \label{eqn:systemexp1}\\
\dot{q}_2^{(2)} =\lambda_2-c_{2K}\alpha_{2K}^{(2)}\beta_K,~\dot{q}_K^{(2)}=c_{2K}\alpha_{2K}^{(2)}\beta_K-\mu_2 \alpha_{KT}^{(2)}
\label{eqn:systemexp2}
\end{subnumcases}
We assume $c_{1K}>\mu_1$ and $c_{2K}>\mu_2$, then the admissible region for $(\lambda_1,\lambda_2)$ is $[0,\mu_1]\times [0,\mu_2]$. If both commodities have arrival rates interior to the admissible region, then the network is queue length stable for both commodities under \eqref{eqn:DiffBacklogBP-Single}. However, suppose that commodity 1 is overloaded ($\lambda_1>\mu_1$), Lemma \ref{lem:2commodity} shows that under $c_{1K}/\mu_1 > c_{2K}/\mu_2$, $\lambda_2<\mu_2$ does not guarantee the queue length stability of commodity 2, which explains the instability example given in the introduction (Fig.~\ref{fig:TwoNodeToyExample}) due to the non-existence of equilibrium point for the subsystem \eqref{eqn:systemexp2}.

\begin{lemma}
\label{lem:2commodity}
For the 2-commodity toy system in Fig.~\ref{fig:TwoNodeToyExample}, suppose that $\lambda_1>\mu_1$ and $c_{1K}/\mu_1 > c_{2K}/\mu_2$, then under the \eqref{eqn:systemexp1} and \eqref{eqn:systemexp2}, the subsystem \eqref{eqn:systemexp2} has an equilibrium point if and only if $\lambda_2\in[0,\frac{\mu_1}{c_{1K}}c_{2K}]\subset [0,\mu_2)$.
\end{lemma}

\begin{proof}
To guarantee that the subsystem \eqref{eqn:systemexp2} has an equilibrium point, we need have $\dot{q}_2^{(2)}=0$ and $\dot{q}_K^{(2)}=0$
i.e., $\lambda_2=c_{2K}\alpha_{2K}^{(2)}\beta_K=\mu_2 \alpha_{KT}^{(2)}$. To capture this, we first identify the value of $\beta_K$.
Since $c_{1K}/\mu_1 > c_{2K}/\mu_2$, then when node $K$ saturates, commodity 1 will squeeze out commodity 2 and dominates in node $K$, hence $\alpha_{1K}^{(1)}\rightarrow 1$ as queue backlog at node 1 increases to infinity and $\alpha_{KT}^{(1)}\rightarrow 1$ since there exists positive queue backlog length of commodity 1 in node $K$. Since node $K$ has finite buffer, then when $t\rightarrow \infty$ we should have
$$
\dot{q}_K^{(1)}=c_{1K}\alpha_{1K}^{(1)}\beta_K-\mu_1 \alpha_{KT}^{(1)}=c_{1K}\beta_K-\mu_1=0,
$$
and thus $\beta_K=\mu_1/c_{1K}$. Then we require
$$
\lambda_2=c_{2K}\alpha_{2K}^{(2)}\beta_K=\frac{\mu_1c_{2K}}{c_{1K}}\alpha_{2K}^{(2)}\in\left[0,\frac{\mu_1c_{2K}}{c_{1K}}\right]
$$
to ensure the existence of an equilibrium point, 
and $\left[0,\frac{\mu_1c_{2K}}{c_{1K}}\right] \subset [0,\mu_2)$ because $c_{1K}/\mu_1 > c_{2K}/\mu_2$.
\end{proof}

%

Lemma \ref{lem:2commodity} quantifies the shrinkage of the range of $\lambda_2$ under which backpressure can stabilize commodity 2. The result can be extended to Theorem \ref{prop:buffer-no-allocation}, where $C$ commodities shares a single buffer in the one-hop system as shown in Fig.~\ref{fig:CcommoditiesExample}. Theorem \ref{prop:buffer-no-allocation} identifies the maximum arrival rate of each commodity to guarantee the existence of an equilibrium point for the subsystem of each commodity under the policy based on backpressure, in the form of \eqref{eqn:systemexp2}, when overloading occurs on one commodity. Proof of Theorem \ref{prop:buffer-no-allocation} is similar to that of Lemma \ref{lem:2commodity} above.

\begin{figure}[!htbp]
\centering
\includegraphics[width=0.8\linewidth]{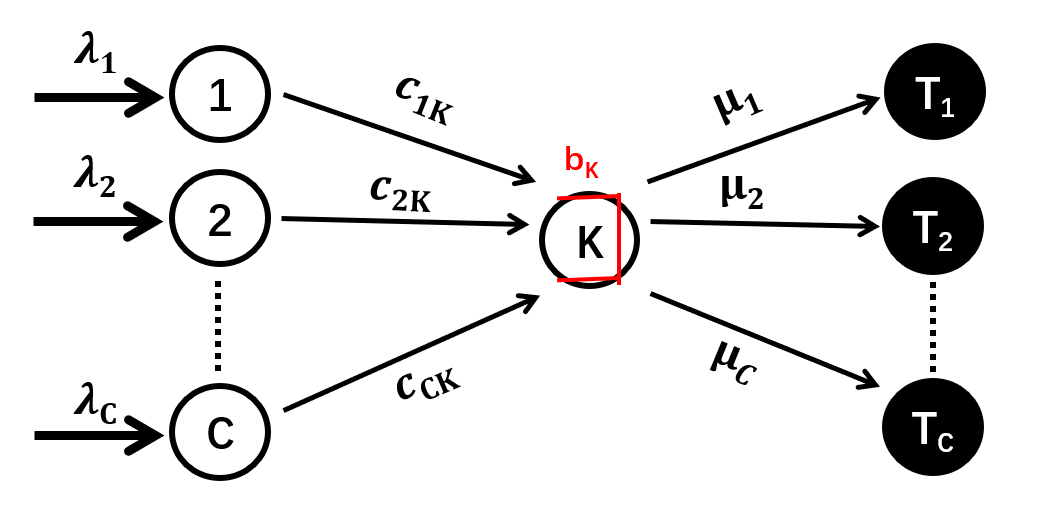}
\caption{One-hop system with $C$ commodities}
\label{fig:CcommoditiesExample}
\end{figure}

\begin{theorem}
\label{prop:buffer-no-allocation}
For the system in Fig.~\ref{fig:CcommoditiesExample}, assume that $c_{iK}>\mu_i,~\forall i=1,2,\dots,C$, and w.l.o.g $c_{iK}/\mu_i>c_{jK}/\mu_j$ for $\forall i,j,~1\leq i<j\leq C$. Suppose that commodity $\ell$ is the \emph{only} overloaded commodity ($\lambda_\ell>\mu_\ell$), then there exists an equilibrium point for the subsystem of commodity $p$ ($p\neq \ell$) with $\lambda_{p}\in [0,\mu_{p})$ for $p=1,2,\dots,\ell-1$ and with $\lambda_{p}\in [0,\mu_\ell c_{p K}/c_{\ell K}]\subset [0,\mu_{p})$ for $p=\ell+1,\dots,C$, under the policy in the form of \eqref{eqn:systemexp2} for every commodity.
\end{theorem}


\section{Conclusion}

In this paper, we propose an ODE model that can capture the dynamics of buffered communication systems to study network stability. For single-commodity systems, we propose a sufficient condition for a local policy to stabilize the network. The result characterizes a set of policies, and captures systems with arbitrary buffers. For such policies the network stability problem is reduced to an problem testing the existence of an equilibrium point for the ODE system. For multi-commodity systems, we extend the condition by incorporating an additional condition on the coupling level between different commodities, and explain the existence of an equilibrium point in different buffer settings.
Future work includes obtaining necessary conditions for network stability, and the analysis of other network performance metrics such as throughput, delay, and fairness using this framework.

\bibliography{FiniteBufferShort}
\bibliographystyle{IEEEtran}

\appendix

\subsection{Proof of Theorem \ref{prop:single1}}
\label{subsec:proof_theorem_1}

The basic idea is to use Theorem \ref{thm:neg_real_part} to show all the eigenvalues of $\mathbf{J}$ at the equilibrium point under \eqref{eqn:local_stability} have negative real part. 
\begin{theorem}
\label{thm:neg_real_part}
\cite{hespanha2018linear} All the eigenvalues of a matrix $\mathbf{J}\in \mathbb{R}^{N\times N}$ have negative real parts if and only if there exists a symmetric positive-definite matrix $\mathbf{A}\in \mathbb{R}^{N\times N}$ so that $\mathbf{A}\mathbf{J} + \mathbf{J}^{T}\mathbf{A}$ is negative definite.
\end{theorem}

We prove the existence of a \emph{diagonal} matrix $\mathbf{A}:=\text{diag}\{a_i\}_{i=1}^N$ with $a_i>0,\forall i=1,\dots,N$ so that $\mathbf{AJ}+\mathbf{J}^T\mathbf{A}$ is negative definite. The proof relies on the notation 
$\mathbf{J}_0 = \mathbf{J} - \boldsymbol{\Lambda}$, 
where $\boldsymbol{\Lambda}:=\text{diag}\left\{-\frac{\partial g_{iT}}{\partial q_i}\right\}_{i=1}^N$, and the following Lemma \ref{lem:positive_vector_exist}, with proof in Appendix \ref{eqn:lemma_positive_vector_exist}.

\begin{lemma}
\label{lem:positive_vector_exist}
$\exists \boldsymbol{\delta}\in \mathbb{R}_+^N$ such that $\mathbf{J}_0\boldsymbol{\delta}=\boldsymbol{0}$. 
\end{lemma}

We now step into the formal proof. Denote $\mathbf{Q}:=\mathbf{AJ}+\mathbf{J}^T\mathbf{A}$, where we can obtain
\begin{equation}
\small
\label{eqn:Q_expression}
\mathbf{Q}_{ij} = \begin{cases}
- a_i \frac{\partial g_{ij}}{\partial q_j} + a_j \frac{\partial g_{ij}}{\partial q_i}, ~i\neq j \\
2a_i\left( \sum_{k: (k,i)\in \mathcal{E}} \frac{\partial g_{ki}}{\partial q_i} -  \sum_{l: (i,l)\in \mathcal{E}} \frac{\partial g_{il}}{\partial q_i} -  \frac{\partial g_{iT}}{\partial q_i} \right),~i=j
\end{cases}
\end{equation}
We show that there exists a negative vector $\boldsymbol{\alpha}:=\{\alpha_{ij}\}_{(i,j)\in \mathcal{E}}$ so that
\begin{equation}
\small
\label{eqn:quadratic_form}
\mathbf{z}^T \mathbf{Qz} = \sum_{i: (i,j)\in \mathcal{E}} \alpha_{ij}\left(a_iz_i - a_jz_j\right)^2 - 2\sum_{i=1}^N a_i \frac{\partial g_{iT}}{\partial q_i} z_i^2
\end{equation}
where $\alpha_{ij}$ is a function of the entries in $\mathbf{J}$. Once \eqref{eqn:quadratic_form} is proved, then $\forall \mathbf{z}  \in \mathbb{R}^N\backslash \boldsymbol{0}$, $\mathbf{z}^T \mathbf{Qz}<0$ due to the existence of $i$ such that $\frac{\partial g_{iT}}{\partial q_i}$ is negative.

Based on Lemma \ref{lem:positive_vector_exist}, there exists $\boldsymbol{\delta}\in \mathbb{R}_+^{N}$ so that $\mathbf{J}_0\boldsymbol{\delta}=\boldsymbol{0}$. Take $\mathbf{A}:=\text{diag}\{\delta_i^{-1}\}_{i=1}^N$, i.e. $a_i=\delta_i^{-1}>0,~\forall i=1,\dots,N$, and take 
$
\alpha_{ij}:=-\delta_i \frac{\partial g_{ij}}{\partial q_i} + \delta_j \frac{\partial g_{ij}}{\partial q_j},
$
which is negative under the condition $\frac{\partial g_{ij}}{\partial q_i}>0,~\frac{\partial g_{ij}}{\partial q_j}<0$. Then based on \eqref{eqn:Q_expression},  we can show the LHS of \eqref{eqn:quadratic_form} is
\begin{equation}
\small
\label{eqn:quadratic_form_LHS}
\begin{aligned}
&\mathbf{z}^T \mathbf{Qz} = \sum_{i=1}^N \mathbf{Q}_{ii} z_i^2 + 2\sum_{1\leq i < j \leq N} \mathbf{Q}_{ij}z_iz_j
\\& = \sum_{i=1}^N \frac{2}{\delta_i} \left(\sum_{k:(k,i)\in \mathcal{E}} \frac{\partial g_{ki}}{\partial q_i} - \sum_{j:(i,j)\in \mathcal{E}} \frac{\partial g_{ij}}{\partial q_i} \right) z_i^2 
\\& \quad + 2\sum_{1\leq i < j \leq N} \left(\frac{1}{\delta_i} \frac{\partial g_{ij}}{\partial q_j} -  \frac{1}{\delta_j} \frac{\partial g_{ij}}{\partial q_i}  \right) z_iz_j - 2\sum_{i=1}^N a_i \frac{\partial g_{iT}}{\partial q_i} z_i^2
\\& = \sum_{i=1}^N \frac{1}{\delta_i^2} \left(\sum_{k:(k,i)\in \mathcal{E}} \alpha_{ki} + \sum_{j:(i,j) \in \mathcal{E}} \alpha_{ij}\right) z_i^2 
\\& \quad + 2\sum_{1\leq i<j\leq N} \alpha_{ij} z_iz_j - 2\sum_{i=1}^N a_i \frac{\partial g_{iT}}{\partial q_i} z_i^2
\end{aligned}
\end{equation}
where the equivalence of the coefficients of $z_i^2$ for any $i$ in the last equation is due to $\mathbf{J}_0\boldsymbol{\delta}=\boldsymbol{0}$. It is  trivial to verify the coefficients of $\{z_i^2\}_{i=1}^N$ and $\{z_{ij}\}_{1\leq i<j\leq N}$ in \eqref{eqn:quadratic_form_LHS} match to the RHS of \eqref{eqn:quadratic_form}, where for $(i,j)\notin \mathcal{E}$, $g_{ij}\equiv 0$ and the coefficient of $z_iz_j$ is $0$. Thus we obtain the proof.

\subsection{Proof of Lemma \ref{lem:positive_vector_exist}}
\label{eqn:lemma_positive_vector_exist}

Note that $\mathbf{J}_0$ is not full rank as we can verify $\boldsymbol{1}^T\mathbf{J}_0=\boldsymbol{0}$. Therefore there exists an eigenvalue $0$. Denote the eigenvalues of $\mathbf{J}_0$ as $\{\lambda_i\}_{i=1}^N$, without loss of generality that they are sorted with non-increasing real part: $\text{Re}(\lambda_1) \geq \text{Re}(\lambda_2) \dots \geq \text{Re}(\lambda_N)$. Since $\mathbf{J}_0$ is diagonally dominant and all diagonal entries are negative, thus $\text{Re}(\lambda_i)\leq 0$ by the Gershgorin circle theorem \cite{li2019eigenvalue}. Therefore $\lambda_1=0$. Suppose $\lambda_i:=r_i+\mathbf{j}u_i,~\forall i\neq 1$ where $r_i\leq 0$.

Denote $\mathbf{J}_{\theta} = \mathbf{J}_0 + \theta \mathbf{I}$. Since all diagonal entries of $\mathbf{J}_0$ are negative while all the off-diagonal entries are non-negative. Then $\exists \theta>0$ so that $\mathbf{J}_{\theta}$ is a matrix with all entries non-negative. Since the network is acyclic
, it does not contain any strongly connected component. Therefore $\mathbf{J}_{\theta}$ is a irreducible non-negative matrix \cite{meyer2000matrix}, to which the Perron-Frobenius theorem \cite{meyer2000matrix} can be applied: For such $\theta$, there exists a real eigenvalue $r>0$ such that (i) $\exists \mathbf{v} \in \mathbb{R}_+^N$ so that $\mathbf{J}_{\theta}\mathbf{v} = r\mathbf{v}$, and (ii) any other eigenvalue of $\mathbf{J}_{\theta}$ has smaller magnitude than $r$.

The $i$-th eigenvalue of $\mathbf{J}_{\theta}$, denoted as $\tilde{\lambda}_i$, equals to $\lambda_i+\theta$. 
Thus $\tilde{\lambda}_1 =\theta$. Take $\theta$ so that $\mathbf{J}_{\theta}$ is non-negative.  

\textbf{Step 1: Consider the case where $\mathbf{J}_0$ has no pure imaginary eigenvalues.} In this case, $r_i<0,~\forall i\neq 1$. The eigenvalues of $\mathbf{J}_{\theta}$ is $\tilde{\lambda}_1 = \theta,~\tilde{\lambda}_i = r_i+\theta + \mathbf{j} u_i$. It is easy to verify that by taking $\theta > \max_{i:i\neq 1} \left\{-2\frac{|\lambda_i|}{\text{Re}(\lambda_i)}\right\}$, we can guarantee that $\theta = \tilde{\lambda}_1 > |\tilde{\lambda}_i|,~\forall i\neq 1$, which means $\tilde{\lambda}_1=\theta = r$.  Thus the eigenvector, denoted as $\mathbf{v}$, of $\mathbf{J}_{\theta}$ associated with eigenvalue $r$, is positive. Therefore
$
\mathbf{J}_{\theta} \mathbf{v} = \left(\mathbf{J}_0 + \theta \mathbf{I}\right)\mathbf{v} = r \mathbf{v} = \theta \mathbf{v}
$
Thus $\mathbf{J}_0\mathbf{v} =0$. This $\mathbf{v}$ is the $\boldsymbol{\delta}$ that meet our condition.

\textbf{Step 2: Prove that $\mathbf{J}_0$ does not have pure imaginary eigenvalues}. It is to show there is no $i\neq 1$ such that $r_i=0$ and $u_i\neq 0$. We prove by contradiction. Suppose that there exists a pair of pure imaginary eigenvalues, w.l.o.g., $\lambda_2$ and $\lambda_3$. Then there exists large enough $\theta>0$ such that $\mathbf{J}_{\theta}$ is non-negative, and $|\tilde{\lambda}_3|=|\tilde{\lambda}_2|>|\tilde{\lambda}_1|>|\tilde{\lambda}_i|,~i\geq 4$. By Perron-Frobenius theorem, $\tilde{\lambda}_3=\lambda_3+\theta$ is a real positive eigenvalue of $\mathbf{J}_{\theta}$, which contradicts that $\lambda_3$ is pure imaginary.

\vspace{1mm}

\subsection{Example of Lemma \ref{prop:single3}}
\label{sec:lemma}

For the policy \eqref{eqn:DiffBacklogBP-Single} based on backpressure,
if we suppose that
$
\lambda_i+\sum_{k:(k,i)\in \mathcal{E}} c_{ki}\leq \sum_{j:(i,j)\in \mathcal{E}} c_{ij}+c_{iT},
$
then we set $\underline{b}_i=0,~\forall i\in \mathcal{V}$, and then set the values of $\{\bar{b}_i\}_{i=1}^N$ such that (i) $\bar{b}_j>\bar{b}_k$ for any link $(j,k)$, and (ii) $\bar{b}_i\leq b_i-\delta_i$ for any node $i$ with finite buffer, where $\delta_i$ is a positive constant close to $0$. We can verify that for $\dot{q}_i=f_i(\mathbf{q})$, when $q_i=\underline{b}_i=0$, then
$
f_i(\mathbf{q})=\lambda_i+\sum_{k:(k,i)\in \mathcal{E}} g_{ki}(q_k,0)-\sum_{j:(i,j)\in \mathcal{E}} g_{ij}(0,q_j)-g_{iT}(0)=\lambda_i+\sum_{k:(k,i)\in \mathcal{E}} g_{ki}(q_k,0)\geq 0,
$
and when $q_i=\bar{b}_i$, note that for any node $j$ such that $(i,j)\in \mathcal{E}$, $g_{ij}(\bar{b}_i,q_j)=c_{ij}$ since $q_j\in [0,\bar{b}_j]\in [0,b_j)\cap[0,\bar{b}_i)$, therefore
$
f_i(\mathbf{q})=\lambda_i+\sum_{k:(k,i)\in \mathcal{E}} g_{ki}(q_k,\bar{b}_i)-\sum_{j:(i,j)\in \mathcal{E}} g_{ij}(\bar{b}_i,q_j)-g_{iT}(\bar{b}_i)
\leq \lambda_i+\sum_{k:(k,i)\in \mathcal{E}}c_{ki}-\sum_{j:(i,j)\in \mathcal{E}} c_{ij}-c_{iT}\leq 0.
$
Therefore the conditions in Lemma \ref{prop:single3} holds, which implies the policy \eqref{eqn:DiffBacklogBP-Single} can ensure a unique stable equilibrium point, and thus render the networks to be queue length stable.

\end{document}